\documentclass{article}
\usepackage{ijcai19}
\usepackage{amsfonts}

\usepackage{times}
\usepackage{soul}
\usepackage{url}
\usepackage[hidelinks]{hyperref}
\usepackage[utf8]{inputenc}
\usepackage[small]{caption}
\usepackage{graphicx}
\usepackage{amsmath}
\usepackage{booktabs}
\usepackage{algorithm}
\usepackage{algorithmic}
\usepackage{tikz}
\usepackage{subfig}
\usetikzlibrary{patterns}
\usepackage{amsthm}

\newtheorem{theorem}{Theorem}
\newtheorem{lemma}{Lemma}
\newtheorem{proposition}{Proposition}
\newtheorem{definition}{Definition}

\title{Mechanism Design for Maximum Vectors}

\author{
Eric Angel$^1$\footnote{Contact Author}\and
Evripidis Bampis$^2$\\
\affiliations
$^1$Univ Evry, Universit\'e Paris-Saclay, IBISC, Evry, France\\
$^2$Sorbonne Universit\'e, UMR 7606, LIP6, Paris, France\\
\emails
eric.angel@univ-evry.fr,
evripidis.bampis@lip6.fr
}

\begin{document}

\maketitle

\begin{abstract}
We consider the {\em Maximum Vectors problem} in a strategic setting. In the classical setting this problem consists, given a set of $k$-dimensional vectors, in computing the set of all \emph{nondominated} vectors. Recall that a vector $v=(v^1, v^2, \ldots, v^k)$ is said to be nondominated if there does
not e\-xist another vector $v_*=(v_*^1, v_*^2, \ldots, v_*^k)$ such that
$v^l \leq v_*^{l}$ for $1\leq l\leq k$, 
with at least one strict inequality among the $k$ inequalities.
This problem is strongly related to other known problems such as the \emph{Pareto curve computation} in multiobjective optimization. In a strategic setting each vector is owned by a selfish agent which can misreport her values in order to become nondominated by other vectors. Our work explores under which conditions it is
possible to incentivize agents to report their true values
using the algorithmic mechanism design framework. We provide both impossibility results along with positive ones, according to various assumptions.    
\end{abstract}


\section{Introduction}
A great variety of algorithms and methods have been designed for various optimization problems. In classic \emph{Combinatorial Optimization}, the algorithm knows the complete  input of the problem, and its goal is to produce  an optimal or near optimal solution. However, in many modern applications the input of the problem is spread among  a set of \emph{selfish agents}, where eachone owns a different part of the input as private knowledge. Hence, every agent is capable to manipulate the algorithm by miss-reporting its part of the input in order to maximize its personal payoff. In their seminal paper Nisan and Ronen  \cite{NR99} were the first to study the impact of the ``strategic'' behavior of the agents on the difficulty of an optimization problem.
Since then  \emph{Algorithmic Mechanism Design} studies optimization problems in presence of selfish agents with private knowledge of the input and potentially conflicting individual objective functions. The goal is to know whether it is possible to propose a \emph{truthful (or incentive compatible) mechanism}, i.e.,  an algorithm solving the optimization problem together with a set of incentives/payments for the agents  motivating them to report honestly their part of the input.  

As an illustrating example, consider the problem of finding the maximum of a set of values $v_1,v_2,\ldots,v_n$. In the classic setting, computing the maximum value is trivial. Let us consider now the case when the inputs are strategic. It means that 
each of the $n$ selfish agents $i$, for $1\leq i\leq n$,
has a private value $v_i$ (not known to the algorithm) for
being selected (as the maximum), and may report any value $b_i$. If the agents know that the maximum value will be computed using the classic algorithm: $\max_i b_i$, then each agent will have an incentive
to cheat by declaring $+\infty$ instead of her true value.

In such a strategic setting, we need a mechanism that is capable to give incentive to the agents to report eachone their true value. For doing that we may use Vickrey's (also known as \emph{second-price}) mechanism \cite{Vickrey}. In this setting, the maximum $\max_i b_i$ is still
computed, but the agent $i^*$=arg max $b_i$ is charged the second
highest reported value $p^*:=\max_{j \neq i^*} b_j$. Her utility is
therefore $v_{i^*}-p^*$. It can be shown that in such a setting each
agent $i$ will have an incentive to report $b_i=v_i$, and henceforth
this mechanism is able to compute the maximum value in this strategic environment \cite{D16}.

In this paper, we consider the problem of \emph{maximum vectors}, i.e., the problem of finding the maxima of a set of vectors in a strategic environment. The classic problem of computing the maxima of a
set of vectors can be stated as follows: we are given a set $V$ of $n$ $k$-dimensional vectors $v_1,v_2,\ldots,v_n$ with $v_i=(v_i^1, v_i^2, \ldots, v_i^k)$ for $i=1,2,\ldots,n$. Given two vectors
$v_i=(v_i^1, v_i^2, \ldots, v_i^k)$ and 
$v_j=(v_j^1, v_j^2, \ldots, v_j^k)$, we say that $v_i$ is {\em dominated}
by $v_j$ if $v_i^l \leq v_j^{l}$ for $1\leq l\leq k$, 
with at least one strict inequality among the $k$ inequalities. The problem consists in computing $MAX(V)$, i.e. the set of all \emph{nondominated} vectors among the $n$ given vectors. This problem is related to other known problems as the \emph{Pareto curve computation}  in multiobjective optimization \cite{Ehr00,PY,S86},  the \emph{skyline} problem in databases \cite{KRR,PTFS03}, or the \emph{contour} problem \cite{M74}. 
In a ``strategic'' setting  the problem is as follows:
 there are $n$ selfish agents $1,2,\ldots,n$ and the value of
 agent $i$ is described by a  vector $v_i=(v_i^1, v_i^2, \ldots, v_i^k)$ for being \emph{selected}\footnote{An agent is \emph{selected} if its bid belongs to the set of nondominated vectors.}. The vector $v_i$ is  a \emph{private information} known only by agent $i$. Computing the set of nondominated vectors by using one of the classic algorithms gives incentive to the agents to cheat by declaring $+\infty$ in all the coordinates of their vectors instead of their true values per coordinate. Our work explores under which conditions it is
 possible to incentivize agents to report their true values. In order to
 precisely answer this question, it is useful to distinguish two cases. 
 In the strongest case, the mechanism is able to enforce truthtelling
 for each agent regardless of the reports of the others (\emph{truthfulness}).
 In the second case, the mechanism is able to enforce truthtelling for each
 agent assuming that the others report their true values
 (\emph{equilibria truthfulness}).

\emph{Previous works}
The Artificial Intelligence (AI) community is faced with many real-world problems involving multiple, conflicting and noncommensurate objectives in path planning \cite{Donald,Khouadjia,Quemy}, game search \cite{Dasgupta}, preference-based configuration \cite{Benabbou}, ... Modeling such problems using a single scalar-valued criterion may be problematic (see for instance \cite {Zeleny})  and hence multiobjective approaches have been studied in the AI literature \cite{Hart,Mandow}.
Some multiobjective problems have been considered in the mechanism design framework. However, these works apply a budget approach where instead of computing the set of all Pareto solutions (or an approximation of this set), they consider that among the different criteria, one is the main criterion to be optimized while the others are modeled via budget  constraints \cite{Bilo,Grandoni}. 

Another family of related works concern \emph{auction theory}. 
In the 
classical setting, the item as well as the valuation of the bidders are characterized 
by a scalar  representing the price/value of the item. However, in many situations 
an item is characterized, besides of its price, by quality measures, delivery times, 
technical specifications etc.  In such cases, the valuation of the bidders for the 
item are vectors. Auctions where the item to sell or buy are characterized by a vector 
are known as \emph{multi-attribute auctions} \cite{Bellosta,Bellosta2,Bichler,Bichler2,Bichler3,Branco,Che,Smet,desmet-new}.
In most of these works, a scoring rule is used for combining the values of the different attributes in order to determine the winner of the auction.

\emph{Our contribution} We first show that neither truthfulness nor  equilibria truthfulness are achievable. However, if one assumes that the agents have distinct values in each of the dimensions, we show that it
is possible to design an equilibria truthful mechanism for the {\em Maximum Vector problem}. We also show that the payments that our algorithm computes are the only payments that give this guarantee. In order to go beyond the negative result concerning ties in the valuations of agents, we show that it is possible to get an equilibria truthful mechanism for the
{\em Weakly Maximum Vector problem} in which one looks for weakly nondominated vectors instead of nondominated ones \cite{Ehr00}.

\section{Problem definition}\label{sec-prbdef}
The following definition and notations will be useful in the sequel of the paper. 

\begin{definition} Given two vectors $x,y  \in \mathbb{R}_+^k$ we say that:
\begin{itemize}
\item $x$ {\em weakly dominates} $y$, denoted by $x \succeq y$, iff $x^j \ge y^j$ for all $j \in \{1, \ldots ,k\}$;
\item $x$ {\em dominates} $y$, denoted by $x \succ y$, iff $x \succeq y$ and $x^j > y^j$ holds for at least one coordinate $j \in \{1, \ldots ,k\}$;
\item $x$ {\em strongly dominates} $y$, denoted by $x \gg y$, iff
$x^j > y^j$ holds for all coordinates $j \in \{1, \ldots ,k\}$;
\item $x$ and  $y$ are {\em incomparable}, denoted by $x \sim y$, 
iff there exist two coordinates, say $j$ and $j'$, such that $x^j < y^j$ and $x^{j'}  > y^{j'}$.
\end{itemize}
\end{definition}

We denote by $\mathbb{R}^k_{+}$ (resp. $\mathbb{R}^k_{*+}$)
the set of vectors $v\in \mathbb{R}^k$ such that $v\succ \vec{0}$ (resp. 
$v\gg \vec{0}$), with $\vec{0}:=(0,\ldots ,0)$ the zero vector.

Given a set $F \subset \mathbb{R}^k_+$, as stated before, we denote by $MAX(F)$ the subset of all nondominated vectors, i.e. $MAX(F) := \{v\in F\: :\: \not\exists v_*\in F,  v_*\succ v \}$.
Such a set is composed of pairwise incomparable vectors. 
In a similar way, we will denote by $MIN(F)$ the set
$\{v\in F\: :\: \not\exists v_*\in F,  v\succ v_* \}$. We will also consider the subset of all weakly nondominated vectors, i.e. $WMAX(F) := \{v\in F\: :\: \not\exists v_*\in F,  v_*\gg v \}$.


The Maximum Vector problem has been studied in the classical framework, and the following proposition is known:

\begin{proposition}
(from \cite{KLP75})\label{lem-par} The set $MAX(F)$ can be computed in
$O(|F| \log |F|)$ time for $k=2,3$ and at most
$O(|F| (\log |F|)^{k-2})$ for $k\geq 4$.
\end{proposition}

Following the mechanism design framework~\cite{IntroMD}, we aim to design 
a mechanism, that we call {\em Pareto mechanism}, such that no agent has an incentive to misreport its vector in order to increase her utility.
The set of agents is denoted by $N$. 
Each agent $i$ has a private {\em vector} $v_i=(v_i^1, v_i^2, \ldots, v_i^k)$ representing the agent valuations on $k$ numerical criteria for being selected. In the following, we consider that $k$ is a fixed constant.
We denote by $V$ the set of private vectors. 
Each agent $i$ reports a vector (a bid) $b_i=(b_i^1, b_i^2,\ldots, b_i^k)$. We denote by $B$ the set of all reported vectors. Based on the set of reported vectors, the mechanism computes for each agent $i$ a vector-payment $p_i=(p_i^1,p_i^2,\ldots ,p_i^k)$.
For each agent $i$, if $b_i$ belongs to $MAX(B)$ she has to pay $p_i$ and so her utility is $u_i := v_i-p_i$, while if $b_i$ does not belong to $MAX(B)$ her utility is $u_i=\vec{0}$ (zero vector). 
Since no agent has an incentive to misreport her vector in order to increase her utility, we will be able to correctly compute $MAX(V)$ by computing $MAX(B)$ since we will have $MAX(B)=MAX(V)$.

If we consider $WMAX$ instead of $MAX$ we use the term of a weakly Pareto mechanism.

\section{Preliminaries}

The Pareto mechanism we want to design must satisfy several properties. 

\begin{definition}[multiobjective individual rationality]
A Pareto mechanism satisfies the multiobjective individual rationality (MIR)
constraint iff $u_i \succeq \vec{0}$ for all agents $i$.
\end{definition}

By the MIR constraint, it is always better for an agent to participate in
the mechanism (i.e. reports a vector) than not participating. In the following we will always assume that the mechanism satisfies the MIR constraint.

We want that the Pareto mechanism incentivize agents to report their true values. This leads to the two following formal definitions. 

\begin{definition}[multiobjective truthfulness]
For any fixed set of reported vectors $b_{i'}$, $i'\neq i$, let $u_i$ be agent $i$'s utility if she reports $b_i=v_i$ and let $u'_i$ denotes her 
utility if she reports $b_i\neq v_i$ (the reported vectors of all the 
other agents remaining unchanged). A Pareto mechanism is said to be 
multiobjective {\em truthful} iff 
$u_i \succeq u'_i$ or $u_i \sim u'_i$ for any agent $i \in 
 \{1, \ldots, n\}$.
\end{definition}

\begin{definition}[multiobjective equilibria truthfulness]
As in the previous definition, let $u_i$ be agent $i$'s utility if $b_i=v_i$ and let 
$u'_i$ denotes her utility if $b_i\neq v_i$. A Pareto mechanism is said to be 
multiobjective {\em equilibria truthful} iff 
$u_i \succeq u'_i$ or $u_i \sim u'_i$ for any agent $i \in 
 \{1, \ldots, n\}$, assuming that $b_{i'}=v_{i'}$ for all $i'\neq i$.
\end{definition}

Honestly reporting her valuation is a dominant strategy for any agent if the mechanism is truthful. 
 
We will also need some additional definitions in the context of multicriteria
optimization, along with some technical lemmas. The missing proofs can be found in 
the Appendix Section. In the sequel, all sets $S$ have a finite size.

\begin{definition}\label{def-t1}
Let $S\subset \mathbb{R}^k_{*+}$ be a finite set of $k$-dimensional vectors.
We define the {\em reference points}\footnote{This set is known in multiobjective optimization as the set of \emph{local upper bounds} \cite{Vander}.} of $S$, denoted by ${\cal T}(S)$,
as the minimum subset of $\mathbb{R}^k_{+}$ such that
for any $v\in \mathbb{R}^k_{*+}$ with $v\not\in MAX(S)$,
one has $v\in MAX(S\cup \{v\})$ iff 
$\exists t\in {\cal T}(S)$ such that $v\gg t$.
\end{definition}

Such a set can be easily computed in dimension 2. For $k=2$, an example is depicted in 
Figure~\ref{fig-refpoints}.
Let $S\subset \mathbb{R}^2_{*+}$. By Proposition~\ref{lem-par}, we compute $MAX(S) = \{s_1,\ldots ,s_r\}$, where the solutions
$s_i$, $1\leq i\leq r$, are pairwise incomparable.
Without loss of generality we assume that 
$s_1^1 < s_2^1 < \ldots < s_r^1$ and $s_1^2 > s_2^2 > \ldots > s_r^2$.
Then one has ${\cal T}(S)=\{t_1,\ldots t_{r+1}\}$ 
with $t_1 = (0,s_1^2)$, $t_{r+1} = (s_r^1,0)$
and $t_l=(s_{l-1}^1,s_l^2)$ for $2\leq l\leq r$.
The overall complexity to compute ${\cal T}(S)$ is therefore $O(|S| \log |S|)$ in dimension 2.

The existence and uniqueness of such a set for any dimension follows from 
Proposition~\ref{prop-ref}. 
Let $D_S^j:=\{0\} \cup \{s^j \, : \, s\in S\}$ for $j=1,\ldots ,k$,
and $\Omega_S := D_S^1 \times D_S^2 \times \cdots \times D_S^k$.

\begin{proposition} \label{prop-ref}
For any finite set $S\subset \mathbb{R}^k_{*+}$,
one has ${\cal T}(S) = MIN( \{t \in \Omega_S \, : \, 
\forall s \in S, \; s \not \gg t \}).$
\end{proposition}

Notice that $|\Omega_S| \leq (|S|+1)^k$ and by using Proposition~\ref{lem-par}
we obtain that for any $S\subset \mathbb{R}^k_{*+}$ its set of reference 
points ${\cal T}(S)$ can be computed in polynomial time with respect to
$|S|$ ($k$ is assumed to be a constant).
For example, with $k=3$ and
$S=\{(2,2,2),(1,3,3),(3,1,1)\}$, using Proposition~\ref{prop-ref} one
obtains:
${\cal T}(S)=\{(3,0,0),(2,1,0),(2,0,1),(1,2,0),(1,0,2),(0,3,0),(0,0,3)\}.$

\begin{figure}
\begin{center}
\begin{tikzpicture}[scale=.5]
\tikzstyle{sommet}=[circle,draw,fill=black,minimum size=5pt,inner sep=0pt]
\tikzstyle{rect}=[rectangle,draw,fill=black,minimum size=5pt,inner sep=0pt]
\draw[->] (0,0) -- (8,0);
\draw (7,0) node[below right] {};
\draw[->] (0,0) -- (0,7);
\draw (0,8) node[above] {};
\draw[step=1cm,gray,thin,dotted] (0,0) grid (7,6);
\draw (1,5) node[sommet]{};
\node[above right] at (1,5) {$s_1$};
\draw (3,4) node[sommet]{};
\node[above right] at (3,4) {$s_2$};
\draw (4,2) node[sommet]{};
\node[above right] at (4,2) {$s_3$};
\draw (6,1) node[sommet]{};
\node[above right] at (6,1) {$s_4$};
\draw (0,5) node[rect]{};
\node[above right] at (0,5) {$t_1$};
\draw (1,4) node[rect]{};
\node[above right] at (1,4) {$t_2$};
\draw (3,2) node[rect]{};
\node[above right] at (3,2) {$t_3$};
\draw (4,1) node[rect]{};
\node[above right] at (4,1) {$t_4$};
\draw (6,0) node[rect]{};
\node[above right] at (6,0) {$t_5$};
\end{tikzpicture}
\caption{\label{fig-refpoints}The reference points in dimension 2. One has 
$S=\{s_1,s_2,s_3,s_4\}$ and ${\cal T}(S)=\{t_1,t_2,t_3,t_4,t_5\}$.}
\end{center}
\end{figure}
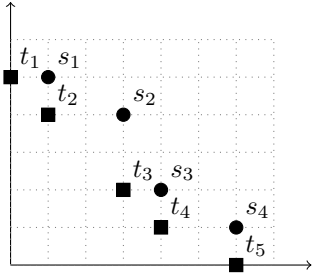

\section{Impossibility results}
Because of the following Proposition, achieving an equilibria truthful Pareto 
mechanism is the best we can hope for.

\begin{proposition}\label{not-truth}A Pareto mechanism cannot be truthful.
\end{proposition}
\begin{proof}
Let us consider an instance in two dimensions, with three agents reporting $b_1=(3,1)$, 
$b_2=(1,3)$ and $b_3=(2,2)$. Then $b_3\in MAX(B)$, and agent 3 is charged
some payment $p_3\preceq (2,2)$ by the MIR assumption.
Let us define the following region of payments ${\cal R} = ([1,2]\times [1,2])
\setminus \{(1,\alpha)\cup (\alpha,1)\}$ for $1\leq \alpha\leq 2$, depicted 
in Figure~\ref{fig-just}.
We claim that we can not have $p_3\in {\cal R}$.
Indeed if it was the case, then if $v_3=(2,2)$ agent 3's interest would
be to lie and report $b_3=((1+p_3^1)/2,(1+p_3^2)/2)$. She would still be in $MAX(B)$ and get charged $p'_3\preceq b_3\prec p_3$.\\
Now, since $p_3\preceq (2,2)$ and $p_3\not\in {\cal R}$ it means that
either $p_3\preceq (1,2)$ (case 1) or $p_3\preceq (2,1)$ (case 2).
In the first case, if $v_3=(1,2.5)$ then agent 3's interest would
be to lie and report $b_3=(2,2)$. She would belong to $MAX(B)$ and
her utility would be $(1,2.5)-p_3\succ (0,0)$, whereas if she reports $b_3=v_3$
she would not belong to $MAX(B)$ and therefore gets a utility $(0,0)$.
In the second case, in a similar way if $v_3=(2.5,1)$ then agent 3's interest 
would be to lie and report $b_3=(2,2)$. 
\end{proof}

\begin{definition} 
An instance satisfies the {\em DV} property (distinct values) if
for every couple of distinct agents $i$, $i' \in N$,  and
every $j \in \{1,\ldots ,k\}$, $v_i^j \neq v_{i'}^j$ holds.
\end{definition} 

Let us motivate the introduction of this property.

\begin{proposition}\label{dvprop} Without the DV property, a Pareto mechanism cannot be 
equilibria-truthful.
\end{proposition}
\begin{proof} 
The proof is very similar to the one for Proposition~\ref{not-truth}.
We only need to assume that agents 1 and 2 are reporting their true vectors, i.e.
$b_1=v_1=(3,1)$ and $b_2=v_2=(1,3)$, and notice that the DV assumption does
not hold in cases 1 and 2.  
\end{proof}


\begin{figure}[t]
\begin{center}
\begin{tikzpicture}[scale=.95]
\tikzstyle{sommet}=[circle,draw,fill=black,minimum size=5pt,inner sep=0pt]
\tikzstyle{rect}=[rectangle,draw,fill=black,minimum size=5pt,inner sep=0pt]
\draw[->] (0,0) -- (4,0);
\draw (6,0) node[below right] {};
\draw[->] (0,0) -- (0,4);
\draw (0,6) node[above] {};
\draw[step=1cm,gray,thin,dotted] (0,0) grid (3,3);
\draw (3,1) node[sommet]{};
\node[above right] at (3,1) {$b_1$};
\draw (1,3) node[sommet]{};
\node[above right] at (1,3) {$b_2$};
\node[above right] at (2,2) {$b_3$};
\draw (2,2) node[sommet] (B) {};

\node[left] at (1,2.5) {$v_3$};
\node[below] at (2.5,1) {$v_3$};
\draw (1,2.5) node[sommet] (A1) {};
\draw (2.5,1) node[sommet] (A2) {};
\draw [->] (A1) to [bend left]  node[midway,above,scale=.7]{case 1} (B);
\draw [->] (A2) to [bend right] node[midway,right,scale=.7]{case 2} (B);

\draw[dashed,red,pattern=north east lines,pattern color=red] (1,1) -- 
(2,1) -- (2,2) -- (1,2) -- (1,1);
\draw[color=red] (1,2) -- (2,2);
\draw[color=red] (2,1) -- (2,2);

\node[below] at (1,0) {$1$}; \node[below] at (2,0) {$2$};
\node[below] at (3,0) {$3$};
\node[left] at (0,1) {$1$}; \node[left] at (0,2) {$2$};
\node[left] at (0,3) {$3$};

\end{tikzpicture}
\caption{\small \label{fig-just}Illustration of the proof of Proposition~\ref{not-truth}.}
\end{center}
\end{figure}
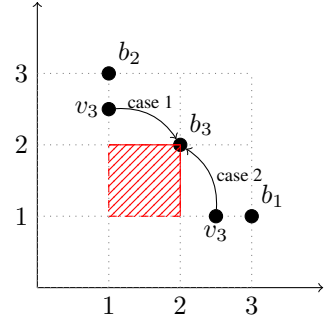

\section{A Pareto mechanism for the Maximum  Vector  problem}
We are going to present a Pareto mechanism, denoted by $\cal M$, which 
satisfies the MIR constraint and which is equilibria-truthful under the hypothesis DV.
The mechanism is described in Algorithm~\ref{figpmeca}. In the initial step, 
we remove all identical vectors. This means that if there is a set of 
agents with the same reported vector $b$, this vector is removed from the set $B$ and
all such agents will not be considered anymore in the mechanism.
Notice that this case will not occur, since we are in the context of
a equilibria truthfulm mechanism and we have the DV assumption. Not having
two identical vectors is a formal requirement used in the proof of 
Lemma~\ref{lem-prop1}.
The mechanism computes for all agents $i$ such that $b_i\in MAX(B)$ a set of possible payments, denoted by $PAY(i)$, and can charge agent $i$ any payment from this set.

\begin{algorithm}
\begin{algorithmic}[1]
\STATE Remove all identical vectors and corresponding agents.
\STATE Compute $MAX(B)$.
\STATE For all $b_i\in MAX(B)$, set $PAY(i) := \{t \in  
{\cal T}(B\setminus \{b_i\}) \: |\: b_i \succeq t\}$, and choose any 
 $p_i\in PAY(i)$.
\STATE For all $b_i\notin MAX(B)$, we set $p_i=\vec{0}$.
\caption{\label{figpmeca}The Pareto mechanism ${\cal M}$.}
\end{algorithmic}
\end{algorithm}

\begin{lemma} \label{lem-prop1}
For any agent $i$ such that $b_i\in MAX(B)$, one has $PAY(i)\neq  \emptyset$.
\end{lemma}
\begin{proof} 
We need the following notations.
Let $D_B^j:=\{0\} \cup \{b^j \, : \, b\in B\}$ for $j=1,\ldots ,k$.
Given $a \in \mathbb{R}_{*+}$, let us denote by
$dec_{B,j}(a)$ the quantity $\max\{x \in D_B^j \, : \, x<a\}$.\\
Let us consider a vector $r =(dec_{B,1}(b_i^1),\ldots,dec_{B,k}(b_i^k))$.
If there exists an agent $z\neq i$ such that $b_z \gg r$ then
$b_z \succ b_i$ (since no two reported vectors are identical), which is a 
contradiction with $b_i \in MAX(B)$.
Therefore $r \in \{t \in \Omega_{B\setminus\{b_i\}} \, : \, \forall s \in B
 \setminus \{b_i\}, \, s \not \gg t\}$.
By definition of the $MIN$ operator, there exists
$w \in MIN(\{t \in \Omega_{B\setminus\{b_i\}} \, : \, \forall s \in B
 \setminus \{b_i\}, \, s \not \gg t\}) = 
{\cal T}(B \setminus \{b_i\})$ such that $w \preceq r$. 
Since $r \prec b_i$ (by construction), we get $w \prec b_i$. 
Therefore $w \in PAY(i)$ and $PAY(i)\neq \emptyset$.
\end{proof}

\begin{theorem} The Pareto mechanism ${\cal M}$ satisfies the MIR constraint.
\label{theo-mir}
\end{theorem}
\begin{proof} 
Given an agent $i$ such that $b_i\not\in MAX(B)$ her utility $u_i$ is the zero vector $\vec{0}$ by definition. 
Now given an agent $i$ such that $b_i\in MAX(B)$, we have to prove that $b_i \succeq p_i$ for all
$p_i \in PAY(i)$.  By Lemma~\ref{lem-prop1}, $PAY(i)$ is nonempty, and according
to the mechanism $\cal M$, $PAY(i)$ contains vectors dominated by $b_i$, thus 
the property holds.
\end{proof}

In what follows, we use the following standard notation in game theory. 
Given a set of reported vectors $B := \{b_j\:|\:j\in N\}$, we denote by 
$(b_{-i},v_i)$ the set in which each agent $j\in N\setminus \{i\}$
reports $b_j$, excepted the agent $i$ which reports $v_i$ instead, and we denote by
$(b_{-i},b_i)$ the set in which each agent $j\in N$ reports $b_j$ 
including the agent $i$ which reports $b_i$.

Before proving Theorem~\ref{theo-eqt} we need the following two lemmas:

\begin{lemma}\label{lem-la}
Let $S\subset \mathbb{R}^k_{*+}$ be a finite set.
Then, $\forall t\in {\cal T}(S)$ and $\forall j\in \{1,\ldots k\}$, then
$t^j=0$ or $\exists s\in S$ such that $t^j=s^j$.
\end{lemma}

\begin{lemma}\label{lem-lt}
Let $S\subset \mathbb{R}^k_{*+}$ be a finite set.
Then, ${\cal T}(S)$ is composed of mutually noncomparable vectors, i.e.
$\forall t,t'\in {\cal T}(S)$, one has $t=t'$ or $t\sim t'$.
\end{lemma}

\begin{theorem}\label{theo-eqt}The Pareto mechanism ${\cal M}$ is equilibria-truthful.
\end{theorem}
\begin{proof}
Let $u_i$ be agent $i$'s utility if $b_i=v_i$ and let $u'_i$ denotes her 
utility if $b_i\neq v_i$.
We need to show that $u_i \succeq u'_i$ or $u_i \sim u'_i$.
Recall that we always assume that $b_i \gg \vec{0}$ and $v_i \gg \vec{0}$.\\
We have two cases to consider. First, let assume that 
$v_i\not\in MAX(b_{-i}, v_i)$.
The utility $u_i$ of agent $i$ is $\vec{0}$.  In that case, agent $i$ may have an 
incentive to report a vector $b_i\neq v_i$ such that $b_i\in MAX(b_{-i},b_i)$. 
According to the mechanism $\cal M$, agent $i$ will be charged some
$t\in {\cal T}(B\setminus \{b_i\})$.
Since $v_i\not\in MAX(b_{-i}, v_i)$ we get from 
Definition~\ref{def-t1} with $S=B\setminus \{b_i\}$ and $v=v_i$
that $v_i\not\gg t$. From the DV (distinct values) hypothesis
and Lemma~\ref{lem-la} and using that $v_i\gg \vec{0}$, we can conclude 
that either $v_i\sim t$ t or $t\gg v_i$. Indeed, if it was not the case, then 
$v_i\not\gg t$, $v_i\not\sim t$ t and $t\not\gg v_i$ implies that 
$\exists j$ such that $v_i^j = t^j$ and by Lemma~\ref{lem-la} we know that
either $t^j=0$ or $\exists s\in B\setminus \{b_i\}$ such that $t^j=s^j$, meaning that
either $v_i^j=0$ or $v_i^j=s^j$ with $s$ the reported vector of a agent different from $i$.
But recall that we have assumed that $v_i \gg \vec{0}$ and moreover since the other agents
report their true values and by the DV hypothesis this is not possible.
This case is illustrated in Figure~\ref{fig-eqt-c1}.
Therefore, the utility
$u'_i = v_i - t$ of agent $i$ will satisfy $\vec{0}\sim u'_i$ or $\vec{0}\gg u'_i$.\\
Assume now that $v_i\in MAX(b_{-i}, v_i)$. According to the mechanism
${\cal M}$, if agent $i$ declares her true value, she will be charged some 
$t$ such that $t\in {\cal T}(B\setminus \{b_i\})$. Her utility $u_i$ will 
be $v_i-t$.
If agent $i$ reports $b_i$ such that $b_i\in MAX(b_{-i},b_i)$, then
she will be charged $t'$ for some
$t'\in {\cal T}(B\setminus \{b_i\})$ and her utility $u'_i$ will be
$v_i - t'$. Since by Lemma~\ref{lem-lt} one has $t=t'$ or $t\sim t'$, the
utilitie will satisfy $u'_i=u_i$ or $u'_i\sim u_i$.
This case is illustrated in Figure~\ref{fig-eqt-c2}.
Finally if agent $i$ reports $b_i$ such that $b_i\not\in MAX(b_{-i},b_i)$,
her utility will be zero, i.e. $u'_i=\vec{0}$ whereas $u_i\succeq \vec{0}$ according to
Theorem~\ref{theo-mir}.

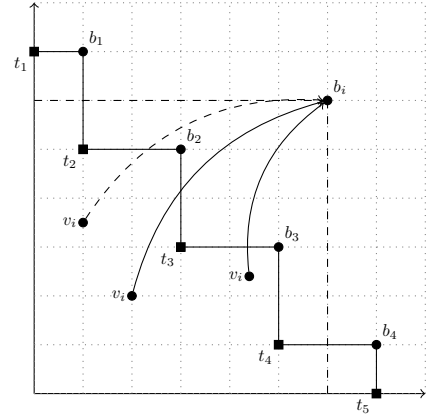
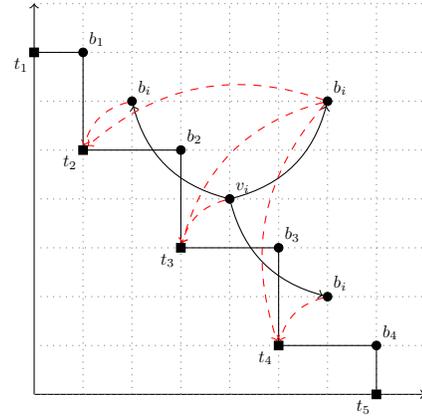
\begin{figure}[t]
\begin{center}
\subfloat[\label{fig-eqt-c1}First Case - Here $t$ can be either $t_2$, $t_3$ 
or $t_4$. The situation corresponding to the dashed line for which one has 
neither $v_i\sim t_2$ nor $t_2\gg v_i$ is not possible.]{
\begin{tikzpicture}[scale=.65, every node/.style={transform shape}]
\tikzstyle{sommet}=[circle,draw,fill=black,minimum size=5pt,inner sep=0pt]
\tikzstyle{rect}=[rectangle,draw,fill=black,minimum size=5pt,inner sep=0pt]
\draw[->] (0,0) -- (8,0);
\draw[->] (0,0) -- (0,8);
\draw[step=1cm,gray,thin,dotted] (0,0) grid (8,8);
\draw (1,7) node[sommet]{}; \node[above right] at (1,7) {$b_1$};
\draw (3,5) node[sommet]{}; \node[above right] at (3,5) {$b_2$};
\draw (5,3) node[sommet]{}; \node[above right] at (5,3) {$b_3$};
\draw (7,1) node[sommet]{}; \node[above right] at (7,1) {$b_4$};
\draw (6,6) node[sommet] (B) {}; \node[above right] at (6,6) {$b_i$};
\draw[dashed] (0,6) -- (6,6); \draw[dashed] (6,0) -- (6,6);
\draw (0,7) node[rect]{}; \node[below left] at (0,7) {$t_1$};
\draw (1,5) node[rect]{}; \node[below left] at (1,5) {$t_2$};
\draw (3,3) node[rect]{}; \node[below left] at (3,3) {$t_3$};
\draw (5,1) node[rect]{}; \node[below left] at (5,1) {$t_4$};
\draw (7,0) node[rect]{}; \node[below left] at (7,0) {$t_5$};
\draw (0,7) -- (1,7) -- (1,5) -- (3,5) -- (3,3) -- (5,3) -- (5,1) -- (7,1) -- (7,0);
\draw (4.4,2.4) node[sommet] (A) {}; \node[left] at (4.4,2.4) {$v_i$};
\draw (2,2) node[sommet] (A2) {}; \node[left] at (2,2) {$v_i$};
\draw (1,3.5) node[sommet] (A3) {}; \node[left] at (1,3.5) {$v_i$};
\draw[->] (A) to [bend left] (B);
\draw[->] (A2) to [bend left] (B);
\draw[dashed,->] (A3) to [bend left] (B);
\end{tikzpicture}
} \ \ \ \ 
\subfloat[\label{fig-eqt-c2}Second Case - The various possible payments 
according to the vector the agent $i$ declares are indicated with a dashed
line.]{
\begin{tikzpicture}[scale=.65, every node/.style={transform shape}]
\tikzstyle{sommet}=[circle,draw,fill=black,minimum size=5pt,inner sep=0pt]
\tikzstyle{rect}=[rectangle,draw,fill=black,minimum size=5pt,inner sep=0pt]
\draw[->] (0,0) -- (8,0);
\draw[->] (0,0) -- (0,8);
\draw[step=1cm,gray,thin,dotted] (0,0) grid (8,8);
\draw (1,7) node[sommet]{}; \node[above right] at (1,7) {$b_1$};
\draw (3,5) node[sommet]{}; \node[above right] at (3,5) {$b_2$};
\draw (5,3) node[sommet]{}; \node[above right] at (5,3) {$b_3$};
\draw (7,1) node[sommet]{}; \node[above right] at (7,1) {$b_4$};
\draw (0,7) node[rect]{}; \node[below left] at (0,7) {$t_1$};
\draw (1,5) node[rect] (T2) {}; \node[below left] at (1,5) {$t_2$};
\draw (3,3) node[rect] (T3) {}; \node[below left] at (3,3) {$t_3$};
\draw (5,1) node[rect] (T4) {}; \node[below left] at (5,1) {$t_4$};
\draw (7,0) node[rect]{}; \node[below left] at (7,0) {$t_5$};
\draw (0,7) -- (1,7) -- (1,5) -- (3,5) -- (3,3) -- (5,3) -- (5,1) -- (7,1) -- (7,0);
\draw (4,4) node[sommet] (B0) {}; \node[above right] at (4,4) {$v_i$};
\draw (2,6) node[sommet] (B1) {}; \node[above right] at (2,6) {$b_i$};
\draw (6,6) node[sommet] (B2) {}; \node[above right] at (6,6) {$b_i$};
\draw (6,2) node[sommet] (B3) {}; \node[above right] at (6,2) {$b_i$};
\draw[->] (B0) to [bend left] (B1);
\draw[->] (B0) to [bend right] (B2);
\draw[->] (B0) to [bend right] (B3);
\draw[red,dashed,->] (B1) to [bend right] (T2);
\draw[red,dashed,->] (B3) to [bend right] (T4);
\draw[red,dashed,->] (B0) to [bend right] (T3);
\draw[red,dashed,->] (B2) to [bend right] (T2);
\draw[red,dashed,->] (B2) to [bend right] (T3);
\draw[red,dashed,->] (B2) to [bend right] (T4);
\end{tikzpicture}
}
\caption{\small \label{fig-eqt}Illustration of the proof of Theorem~\ref{theo-eqt}.}
\end{center}
\end{figure}

\end{proof}

We are now going to prove that the Pareto mechanism we introduced is the
unique way of achieving equilibria truthfulness.

Let $\pi$ be a truthful payment function, i.e. given a set of reported vectors $B$, for any agent $i$, $\pi_i(B)$ is the amount charged to the agent $i$, such that no agent has an incentive to declare a false 
vector. 
Recall that by the MIR constraint, we assume that 
\begin{eqnarray}
b_i \succeq \pi_i(b_{-i},b_i) \mbox{ for any vector } b_i. \label{ineg1}
\end{eqnarray}

\begin{lemma}\label{lem1}
For any two different reported vectors
$b_i$ and $b'_i$, such that $b_i\in MAX(b_{-i}, b_i)$ and 
$b'_i\in MAX(b_{-i}, b'_i)$, either
$\pi_i(b_{-i},b_i)\sim \pi_i(b_{-i},b'_i)$ or
$\pi_i(b_{-i},b_i) = \pi_i(b_{-i},b'_i)$.
\end{lemma}
\begin{proof}
Let us assume that $\pi_i(b_{-i},b'_i)\succ \pi_i(b_{-i},b_i)$. Then
if $v_i=b'_i$, agent $i$ would have an incentive to report $b_i$ instead
of her true value $v_i$.
The same line of reasoning shows that one cannot have
$\pi(b_{-i},b_i) \succ \pi(b_{-i},b'_i)$ neither. \\$\Box$
\end{proof}

We need additional lemmas. 

\begin{lemma}\label{lem-t01}
Let $S\subset \mathbb{R}^k_{*+}$ be a finite set.
Then, $\forall s\in MAX(S)$, $\exists t\in {\cal T}(S)$ such that
$s \succ t$.
\end{lemma}

\begin{lemma}\label{lem-t5}Let $x,y,z \in \mathbb{R}^k_+$ be three
$k$-dimensional vectors, such that $y\succ x$ and $y\sim z$.
Then, $x\not\succeq z$
\end{lemma}
\begin{proof}
The proof is by contradiction. Assume that $x\succeq z$. Since
$y\sim z$, one can assume without loss of generality that
$y^1 > z^1$ and $y^2 < z^2$.
Since $y\succ x$, one has $y^2 \geq x^2$. Therefore, $z^2 > x^2$, which is
in contradiction with $x\succeq z$.
\end{proof}

\begin{lemma}\label{lem-t4}Let $S\subset \mathbb{R}^k_+$ be a set of
$k$-dimensional vectors, mutually non comparable, i.e.
$\forall s,s'\in S$, one has $s\sim s'$.
Let $p\in \mathbb{R}^k_+$ such that $s^*\succ p$, for some $s^*\in S$.
Then $\forall s\in S$, $s\succ p$ or $s\sim p$.
\end{lemma}
\begin{proof}
Let us consider any $s\in S$ with $s\neq s^*$. We know that $s\sim s^*$,
and by using Lemma~\ref{lem-t5} with $x=p$, $y=s^*$ and $z=s$, we obtain
that $p\not\succeq s$, i.e. $s\gg p$ or $s\sim p$.
\end{proof}

\begin{lemma}\label{lem-t2}Let $S\subset \mathbb{R}^k_+$ be a set of
$k$-dimensional vectors, and let $p\in \mathbb{R}^k_+$ such that
$\forall s\in S$, $p\sim s$. Then, there exist $\delta \in \mathbb{R}^k_{*+}$
such that $\forall s\in S$, $p+\delta \sim s$.
\end{lemma}
\begin{proof}
Let $S=\{s_1,\ldots ,s_{|S|}\}$. Since $p\sim s$, $\forall s\in S$, it means
that $\forall s_i\in S$, $\exists l_i,l'_i \in \{1,\ldots ,k\}$ such that
$s_i^{l_i} > p^{l_i}$ and $s_i^{l'_i} < p^{l'_i}$.
Let $I_j = \{ i \: | \: l_i = j \}$.
For all $j\in \{1,\ldots ,k\}$, let 
$\varepsilon^j = \min_{i\in I_j} (s_i^{l_i} - p^{l_i})$ if 
$I_j \neq \emptyset$, and let $\varepsilon^j$ be any positive real number 
otherwise.
Then it is easy to see that the vector
$\delta=(\varepsilon^1/2,\ldots ,\varepsilon^k/2)$ satisfies the lemma.
Indeed, observe that $\varepsilon^j>0$ for $1\leq j\leq k$.
Now let us consider any $s_r\in S$ with $1\leq r\leq |S|$.
One has $s_r^{l_r} > p^{l_r}$ and $s_r^{l'_r} < p^{l'_r}$.
Therefore $s_r^{l'_r} < p^{l'_r} + \varepsilon^{l'_r}/2 = (p+\delta)^{l'_r}$.
We are going to show that $s_r^{l_r} > (p+\delta)^{l_r}$.
One has $r\in I_{l_r}$, hence $\varepsilon_{l_r} \leq s_r^{l_r}-p^{l_r}$, and
$\delta^{l_r}=\varepsilon_{l_r}/2 < s_r^{l_r}-p^{l_r}$, which means that
$s_r^{l_r} > (p+\delta)^{l_r}$.
\end{proof}

\begin{lemma}\label{lem-t3}Let $S\subset \mathbb{R}^k_+$ a set of
$k$-dimensional vectors, and let $p\in \mathbb{R}^k_+$ such that
one has $\forall s\in S_1$, $s\sim p$, and one has
$\forall s\in S_2$, $s\succ p$, with $S_1$ and $S_2$ a bipartition of $S$,
i.e. $S_1\cup S_2=S$ and $S_1\cap S_2=\emptyset$.
Then, there exist $\delta\in \mathbb{R}^k_{*+}$ such that 
$\forall s\in S$, $s\sim p+\delta$ or $s\succ p+\delta$.
\end{lemma}
\begin{proof}
This result can be considered as a generalization of Lemma~\ref{lem-t2} 
and the proof is quite similar. 
For the set $S_1$ we define $\epsilon^j$ as previously. For the set
$S_2$, we define $\varepsilon'^j = \min_{\{s\in S_2\, : \, s^j-p^j>0\}} s^j-p^j$
if the set $\{s\in S_2\, : \, s^j-p^j>0\}$ is not empty, otherwise
$\varepsilon'^j$ is any positive value.
Then it is easy to see that the vector
$\delta$ with $\delta^j = \min \{\varepsilon^j, \varepsilon'^j\}/2$,
for $1\leq j\leq k$ satisfies the lemma.
For example, let us consider any $s\in S_2$.
Since $s\succ p$, $\exists l$, $1\leq l\leq k$, such that $s^l>p^l$.
Then $s^l \geq p^l+\varepsilon'^l > p^l + \delta^l$.
It implies that $s\sim p+\delta$ or $s\succ p+\delta$.
If $s\in S_1$ the proof is the same than Lemma~\ref{lem-t2}.
\end{proof}

\begin{lemma}\label{lem-infinity}
Let $z_j$ and $u_j$, for $j\geq 1$, be two infinite sequences of points in
$\mathbb{R}^k_{+}$ such that $z_j \gg z_{j+1}$, $z_j \succeq u_j\succ t$, and
$\lim_{j\to \infty} z_j = t$. Then $\exists l,l'$ such that $u_l \succ u_{l'}$.
\end{lemma}
\begin{proof}
First, we assume there exist a point $u_i$ such that $u_i\gg t$.
Let $\delta=\min_{1\leq j\leq k} u_i^j - t >0$.
Since $\lim_{j\to \infty} z_j = t$, $\exists N$ such that $\forall i\geq N$
one has $||z_i-t||_{\infty} = \max_{1\leq j\leq k} z_i^j-t^j \leq \delta/2$.
We have $u_i^j \geq t+\delta > t+\delta/2 \geq z_N^j \geq u_N^j$, meaning
that $u_i \succ u_N$.\\
Now assume that $\forall i$, $u_i\not\gg t$.
Since $u_i\succ t$, it means that $\exists \alpha_i$ such that 
$u_i^{\alpha_i}=t^{\alpha_i}$. Without loss of generality by considering 
an (infinite) subsequence of $u_i$ we can assume that $\forall i$, $\alpha_i=1$.
Again by considering an (infinite) subsequence we can assume that
$\exists K$ with $1\leq K\leq k-1$ such that 
$\forall j$, $1\leq j\leq K$, one has $u_i^j =t^j$, and
$\forall j$, $K+1\leq j\leq k$, one has $u_i^j > t^j$. Now we can apply
the same line of reasoning than in the first case, by considering only
the coordinates between $K+1$ and $k$.
\end{proof}

\begin{theorem}\label{theo-all}The Pareto mechanism $\cal M$ gives all the equilibria 
truthful payments.
\end{theorem}
\begin{proof}
The proof is by contradiction. We assume that the payment is computed in a different way than in our Pareto mechanism, and we show that there exists a configuration for which an agent has an incentive to lie. We consider any agent $i$, and any vector $b_i$ such that
$b_i\in MAX(b_{-i}, b_i)$.\\
Let first assume that the payment computed $\pi(b_i)$ is incomparable with 
the set of points ${\cal T}(B\setminus \{b_i\})$, i.e. 
$\forall t\in {\cal T}(B\setminus \{b_i\})$ one has $\pi(b_i)\sim t$.
Using Lemma~\ref{lem-t2}, there exists $\delta\gg \vec{0}$ such that
$\forall t\in {\cal T}(B\setminus \{b_i\})$, $\pi(b_i)+\delta\sim t$.
Let us assume that $v_i=\pi(b_i)+\delta$. Since 
$\forall t\in {\cal T}(B\setminus \{b_i\})$, $v_i\sim t$, we know by
Definition~\ref{def-t1} that 
$v_i\not\in MAX(b_{-i}, v_i)$ or 
$v_i\in MAX(B\setminus \{b_i\})$.
One cannot have $v_i\in MAX(B\setminus \{b_i\})$, because
Lemma~\ref{lem-t01} would contradict the fact that
$\forall t\in {\cal T}(B\setminus \{b_i\})$, $v_i\sim t$.
Therefore one has $v_i\not\in MAX(b_{-i}, v_i)$,
and the utility of agent $i$ is $\vec{0}$ if $i$ reports her true value $v_i$.
If agent $i$ reports $b_i$ her utility will be $u'_i=v_i-\pi(b_i)=\delta\gg \vec{0}$,
meaning that she has an incentive to lie.
This case is illustrated in the Figure~\ref{fig-theo-all} (Case 1).\\
We assume now that the payment computed $\pi(b_i)$ is dominated by
at least one point from ${\cal T}(B\setminus \{b_i\})$.
Using Lemma~\ref{lem-t4} it means that there exists a bipartition 
${\cal T}_1,{\cal T}_2$ of ${\cal T}(B\setminus \{b_i\})$ such that
$\forall t\in {\cal T}_1$, $t \succ \pi(b_i)$ and
$\forall t\in {\cal T}_2$, $t \sim \pi(b_i)$.
Using Lemma~\ref{lem-t3} there exists $\delta\gg \vec{0}$ such that
$\forall t\in {\cal T}(B\setminus \{b_i\})$, $t \succ \pi(b_i)+\delta$ or
$t \sim \pi(b_i)+\delta$.
Let us assume that $v_i=\pi(b_i)+\delta$.  According to the 
Definition~\ref{def-t1}, $v_i\not\in MAX(b_{-i}, v_i)$ or
$v_i\in MAX(B\setminus \{b_i\})$. One cannot have 
$v_i\in MAX(B\setminus \{b_i\})$, because
Lemma~\ref{lem-t01} would contradict the fact that
$\forall t\in {\cal T}_1$, $t \succ v_i$ and
$\forall t\in {\cal T}_2$, $t \sim v_i$.
Therefore $v_i\not\in MAX(b_{-i},v_i)$ and the utility of agent $i$ is 
$\vec{0}$ if $i$ reports her true value $v_i$.
If agent $i$ reports $b_i$ her utility will be $u'_i=v_i-\pi(b_i)=\delta\gg \vec{0}$,
meaning that she has an incentive to lie. This case is illustrated in the
Figure~\ref{fig-theo-all} (Case 2).\\
We assume now that the payment computed $\pi(b_i)$ strictly dominates at least
one point $t$ from ${\cal T}(B\setminus \{b_i\})$.
Now there are two cases to consider, either 
$\pi(\pi(b_i))\neq \pi(b_i)$ or $\pi(\pi(b_i))=\pi(b_i)$.
If $\pi(\pi(b_i))\neq \pi(b_i)$, by using (\ref{ineg1}) one has
$\pi(b_i)\succ \pi(\pi(b_i))$.
If we assume that $v_i=b_i$ then agent $i$ would have an incentive to report 
$\pi(b_i)$ instead of her true value $b_i$. Indeed, since $\pi(b_i)\gg t$,
one can conlude from Definition~\ref{def-t1} that
$\pi(b_i)\in MAX(b_{-i},\pi(b_i))$.
Moreover, agent $i$ would pay $\pi(\pi(b_i))$ instead of $\pi(b_i)$. This case
is illustrated in the Figure~\ref{fig-theo-all} (Case 3).
Assume now that $\pi(\pi(b_i))=\pi(b_i)$. Consider any vector $b'_i$ such that 
$\pi(b_i)\succ b'_i$ and $b'_i\in MAX(b_{-i}, b'_i)$. For example,
one can take $b'_i = \pi(b_i) - (\pi(b_i)-t)/2$ (since $\pi(b_i)\gg t$, one 
has $b'_i\gg t$ and hence by Definition~\ref{def-t1} one has 
$b'_i\in MAX(b_{-i}, b'_i)$).
Using Lemma~\ref{lem1}, either $\pi(b'_i)$ is incomparable with 
$\pi(\pi(b_i))=\pi(b_i)$, or $\pi(b'_i)=\pi(\pi(b_i))=\pi(b_i)$.
If $\pi(b'_i)\sim \pi(b_i)$ then by using Lemma~\ref{lem-t5} with
$x=b'_i$, $y=\pi(b_i)$ and $z=\pi(b'_i)$, one has
$b'_i\not\succeq \pi(b'_i)$. However this contradict the assumption that 
$b'_i\succeq \pi(b'_i)$ (see (\ref{ineg1})).
If $\pi(b'_i) = \pi(b_i)$ then since $\pi(b_i)\succ b'_i$, we have
$\pi(b'_i)\succ b'_i$, and there is again a contradiction with the
assumption (\ref{ineg1}). This case is illustrated in the 
Figure~\ref{fig-theo-all} (Case 4).\\
According to the previous discussion, the only remaining case we need to
consider is when $\pi(b_i)$ dominates, but not strictly dominates, at least 
one point from ${\cal T}(B\setminus \{b_i\})$.
Since $b_i\in MAX(b_{-i}, b_i)$ (and 
$b_i\not\in MAX(B\setminus \{b_i\})$) we know by
Definition~\ref{def-t1} that $\exists t'\in {\cal T}(B\setminus \{b_i\})$
such that $b_i\gg t'$. 
We are going to consider a set of vectors
$b_{i,l}$, with $l\geq 1$, such that $\forall l\geq 1$,
$b_{i,l}\gg b_{i,l+1}$, $b_{i,l}\gg t'$ and $\lim_{l\to \infty} b_{i,l} = t'$.
Since ${\cal T}(B\setminus \{b_i\})$ has a finite number of elements,
we can assume without loss of generality that all the vectors $\pi(b_{i,l})$
dominates the same point $t\in {\cal T}(B\setminus \{b_i\})$.
Recall also that, by inequality~(\ref{ineg1}), 
$b_{i,l} \succeq \pi(b_{-i},b_{i,l})= \pi(b_{i,l})$.
One can easily see that necessarily $t=t'$. Now using Lemma~\ref{lem-infinity}
we obtain that, $\exists l,l'$ such that $\pi(b_{i,l}) \succ \pi(b_{i,l'})$.
If we assume that $v_i=b_{i,l}$ then agent $i$ would have an incentive to report 
$b_{i,l'}$ instead of her true value $b_{i,l}$.
\end{proof}

\begin{figure}[t]
\begin{center}
\subfloat[Case 1]{
\begin{tikzpicture}[scale=.5, every node/.style={transform shape}]
\tikzstyle{sommet}=[circle,draw,fill=black,minimum size=5pt,inner sep=0pt]
\tikzstyle{rect}=[rectangle,draw,fill=black,minimum size=5pt,inner sep=0pt]
\draw[->] (0,0) -- (8,0);
\draw[->] (0,0) -- (0,8);
\draw[step=1cm,gray,thin,dotted] (0,0) grid (8,8);
\draw (1,7) node[sommet]{}; \node[above right] at (1,7) {$b_1$};
\draw (3,5) node[sommet]{}; \node[above right] at (3,5) {$b_2$};
\draw (5,3) node[sommet]{}; \node[above right] at (5,3) {$b_3$};
\draw (7,1) node[sommet]{}; \node[above right] at (7,1) {$b_4$};
\draw (6,6) node[sommet] (B) {}; \node[above right] at (6,6) {$b_i$};
\draw[dashed] (0,6) -- (6,6); \draw[dashed] (6,0) -- (6,6);
\draw (0,7) node[rect]{}; \node[below left] at (0,7) {$t_1$};
\draw (1,5) node[rect]{}; \node[below left] at (1,5) {$t_2$};
\draw (3,3) node[rect]{}; \node[below left] at (3,3) {$t_3$};
\draw (5,1) node[rect]{}; \node[below left] at (5,1) {$t_4$};
\draw (7,0) node[rect]{}; \node[below left] at (7,0) {$t_5$};
\draw (0,7) -- (1,7) -- (1,5) -- (3,5) -- (3,3) -- (5,3) -- (5,1) -- (7,1) -- (7,0);
\draw (3.7,1.7) node[sommet] {}; \node[below left] at (3.7,1.7) {$\pi(b_i)$};
\draw (4.4,2.4) node[sommet] (A) {}; \node[left] at (4.4,2.4) {$v_i = \pi(b_i)+\delta\:$};
\draw[->] (A) to [bend left] (B);
\end{tikzpicture}
}
\subfloat[Case 2]{
\begin{tikzpicture}[scale=.5, every node/.style={transform shape}]
\tikzstyle{sommet}=[circle,draw,fill=black,minimum size=5pt,inner sep=0pt]
\tikzstyle{rect}=[rectangle,draw,fill=black,minimum size=5pt,inner sep=0pt]
\draw[->] (0,0) -- (8,0);
\draw[->] (0,0) -- (0,8);
\draw[step=1cm,gray,thin,dotted] (0,0) grid (8,8);
\draw (1,7) node[sommet] {}; \node[above right] at (1,7) {$b_1$};
\draw (3,5) node[sommet] {}; \node[above right] at (3,5) {$b_2$};
\draw (5,3) node[sommet]{}; \node[above right] at (5,3) {$b_3$};
\draw (7,1) node[sommet]{}; \node[above right] at (7,1) {$b_4$};
\draw (6,6) node[sommet] (B) {}; \node[above right] at (6,6) {$b_i$};
\draw[dashed] (0,6) -- (6,6); \draw[dashed] (6,0) -- (6,6);
\draw (0,7) node[rect] (T1) {}; \node[below left] at (0,7) {$t_1$};
\draw (1,5) node[rect] (T2) {}; \node[below left] at (1,5) {$t_2$};
\draw (3,3) node[rect] (T3) {}; \node[below left] at (3,3) {$t_3$};
\draw (5,1) node[rect] (T4) {}; \node[below left] at (5,1) {$t_4$};
\draw (7,0) node[rect] (T5) {}; \node[below left] at (7,0) {$t_5$};
\draw (0,7) -- (1,7) -- (1,5) -- (3,5) -- (3,3) -- (5,3) -- (5,1) -- (7,1) -- (7,0);
\draw (2,1) node[sommet] {}; \node[below left] at (2,1) {$\pi(b_i)$};
\draw (2.5,1.5) node[sommet] (A) {}; \node[left] at (2.5,1.5) {$v_i = \pi(b_i)+\delta\:$};
\draw[->] (A) to [bend left=10] (B);
\draw[line width=35pt, line cap=round, rounded corners, opacity=0.1] (T1) -- (T2);
\draw[line width=35pt, line cap=round, rounded corners, opacity=0.1] (T3) -- (T4);
\draw[line width=35pt, line cap=round, rounded corners, opacity=0.1] (T5) -- (T5);
\node at (1.5,6.5) {${\cal T}_2$}; \node at (7.5,0.5) {${\cal T}_2$};
\node at (5.5,2) {${\cal T}_1$};
\end{tikzpicture}
}\\
\subfloat[Case 3]{
\begin{tikzpicture}[scale=.5, every node/.style={transform shape}]
\tikzstyle{sommet}=[circle,draw,fill=black,minimum size=5pt,inner sep=0pt]
\tikzstyle{rect}=[rectangle,draw,fill=black,minimum size=5pt,inner sep=0pt]
\draw[->] (0,0) -- (8,0);
\draw[->] (0,0) -- (0,8);
\draw[step=1cm,gray,thin,dotted] (0,0) grid (8,8);
\draw (1,7) node[sommet]{}; \node[above right] at (1,7) {$b_1$};
\draw (5,3) node[sommet]{}; \node[above right] at (5,3) {$b_3$};
\draw (3,5) node[sommet]{}; \node[above right] at (3,5) {$b_2$};
\draw (7,1) node[sommet]{}; \node[above right] at (7,1) {$b_4$};
\draw (6,6) node[sommet] (B) {}; \node[above right] at (6,6) {$v_i=b_i$};
\draw[dashed] (0,6) -- (6,6); \draw[dashed] (6,0) -- (6,6);
\draw (0,7) node[rect]{}; \node[below left] at (0,7) {$t_1$};
\draw (1,5) node[rect]{}; \node[below left] at (1,5) {$t_2$};
\draw (3,3) node[rect]{}; \node[below left] at (3,3) {$t_3$};
\draw (5,1) node[rect]{}; \node[below left] at (5,1) {$t_4$};
\draw (7,0) node[rect]{}; \node[below left] at (7,0) {$t_5$};
\draw (0,7) -- (1,7) -- (1,5) -- (3,5) -- (3,3) -- (5,3) -- (5,1) -- (7,1) -- (7,0);
\draw (4.5,4.5) node[sommet] (pibi) {}; \node[above] at (4.5,4.5) {$\pi(b_i)$};
\draw (2,2) node[sommet] {}; \node[below left] at (2,2) {$\pi(\pi(b_i))$};
\draw[->] (B) to [bend left] (pibi);
\end{tikzpicture}
}
\subfloat[Case 4]{
\begin{tikzpicture}[scale=.5, every node/.style={transform shape}]
\tikzstyle{sommet}=[circle,draw,fill=black,minimum size=5pt,inner sep=0pt]
\tikzstyle{rect}=[rectangle,draw,fill=black,minimum size=5pt,inner sep=0pt]
\draw[->] (0,0) -- (8,0);
\draw[->] (0,0) -- (0,8);
\draw[step=1cm,gray,thin,dotted] (0,0) grid (8,8);
\draw (1,7) node[sommet]{}; \node[above right] at (1,7) {$b_1$};
\draw (5,3) node[sommet]{}; \node[above right] at (5,3) {$b_3$};
\draw (3,5) node[sommet]{}; \node[above right] at (3,5) {$b_2$};
\draw (7,1) node[sommet]{}; \node[above right] at (7,1) {$b_4$};
\draw (6,6) node[sommet] (B) {}; \node[above right] at (6,6) {$b_i$};
\draw[dashed] (0,6) -- (6,6); \draw[dashed] (6,0) -- (6,6);
\draw (0,7) node[rect]{}; \node[below left] at (0,7) {$t_1$};
\draw (1,5) node[rect]{}; \node[below left] at (1,5) {$t_2$};
\draw (3,3) node[rect]{}; \node[below left] at (3,3) {$t_3$};
\draw (5,1) node[rect]{}; \node[below left] at (5,1) {$t_4$};
\draw (7,0) node[rect]{}; \node[below left] at (7,0) {$t_5$};
\draw (0,7) -- (1,7) -- (1,5) -- (3,5) -- (3,3) -- (5,3) -- (5,1) -- (7,1) -- (7,0);
\draw (4.5,4.5) node[sommet] (pibi) {}; \node[right] at (4.5,4.5) {$\:\pi(\pi(b_i))=\pi(b_i)$};
\draw (3.75,3.75) node[sommet] (bip) {}; \node[below] at (3.75,3.75) {$b'_i$};
\draw (4,5) node[sommet] (pibip) {}; \node[above right] at (4,5) {$\pi(b'_i)$};
\end{tikzpicture}
}
\caption{\small \label{fig-theo-all}An illustration of the proof of 
Theorem~\ref{theo-all}.}
\end{center}
\end{figure}
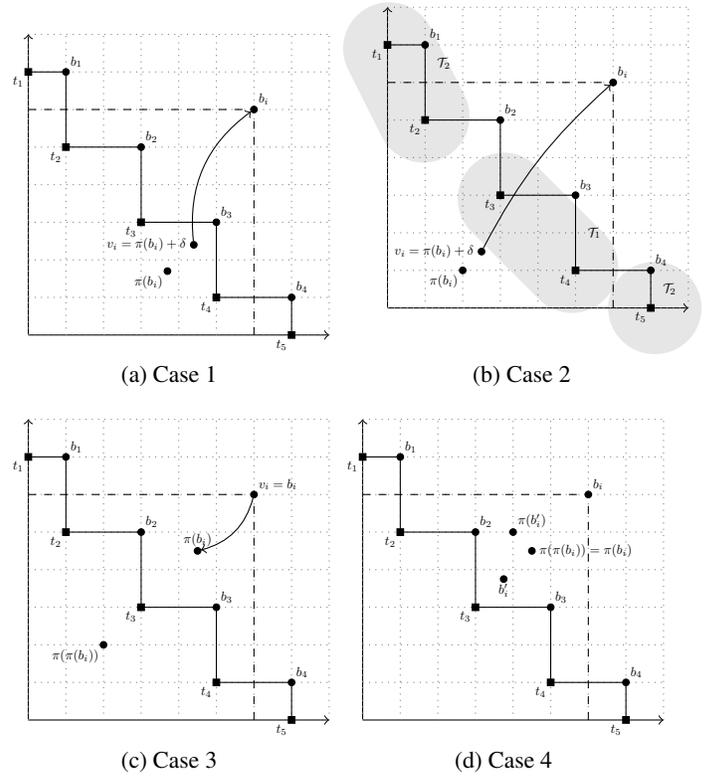

\section{A Pareto mechanism for the Weakly Maximum  Vector  problem}
We are going to present a Pareto mechanism, denoted by $\cal M'$, which 
satisfies the MIR constraint and which is equilibria-truthful.
For doing so, we modify the mechanism ${\cal M}$ in order to
remove the DV condition. 
The modified mechanism is given
in Table~\ref{figpmecarelaxed} and illustrated in 
Figure~\ref{fig-mprime}. It can be shown that $\cup_{i\in W} b_i = WMAX(B)$. 

\bigskip


\begin{algorithm}[h]
\begin{algorithmic}[1]
\STATE Remove all identical vectors and corresponding agents.
\STATE For all $i\in B$, set $PAY(i) := \{t \in  
{\cal T}(MAX(B)\setminus \{b_i\}) \: |\: b_i \succeq t\}$.
\STATE $W$ is the set of agents $i$ such that $PAY(i)\neq \emptyset$.
\STATE For all $i\in W$ choose any $p_i\in PAY(i)$.
\STATE For all $i\notin W$, we set $p_i=\vec{0}$.
\caption{\label{figpmecarelaxed}The weakly Pareto mechanism ${\cal M'}$.}
\end{algorithmic}
\end{algorithm}

\bigskip

\begin{figure}[h]
\begin{center}
{
\begin{tikzpicture}[scale=.8, every node/.style={transform shape}]
\tikzstyle{sommet}=[circle,draw,fill=black,minimum size=5pt,inner sep=0pt]
\tikzstyle{cercle}=[circle,draw,minimum size=9pt,inner sep=0pt]
\tikzstyle{rect}=[rectangle,draw,fill=black,minimum size=5pt,inner sep=0pt]
\draw[->] (0,0) -- (8,0);
\draw[->] (0,0) -- (0,4);
\draw[step=1cm,gray,thin,dotted] (0,0) grid (8,4);
\draw (1,1) node[sommet]{}; \node[below] at (1,1) {$b_1$}; 
\draw (2,1) node[sommet]{}; \node[below] at (2,1) {$b_2$}; 
\draw (4,1) node[sommet]{}; \node[below] at (4,1) {$b_3$}; 
\draw (5,1) node[sommet]{}; \node[below] at (5,1) {$b_4$}; 
\draw (6,1) node[sommet]{}; \node[below] at (6,1) {$b_5$}; 
\draw (6,1) node[cercle]{};
\draw (1,2) node[sommet]{}; \node[left] at (1,2) {$b_6$}; 
\draw (3,2) node[sommet]{}; \node[right] at (3.1,2) {$b_7$};
\draw (3,2) node[cercle]{}; 
\draw (2,3) node[sommet]{}; \node[right] at (2.1,3) {$b_8$}; 
\draw (2,3) node[cercle]{}; 
\draw (0,3) node[rect]{}; \node[above right] at (0,3) {$t_1$}; 
\draw (2,2) node[rect]{}; \node[above right] at (2,2) {$t_2$}; 
\draw (3,1) node[rect]{}; \node[above right] at (3,1) {$t_3$}; 
\draw (6,0) node[rect]{}; \node[above right] at (6,0) {$t_4$}; 
\end{tikzpicture}
} 
\caption{\small \label{fig-mprime}An illustration of the mechanism
${\cal M'}$. One has $MAX(B)=\{b_8,b_7,b_5\}$,
${\cal T}(MAX(B))=\{t_1,t_2,t_3,t_4\}$, and
$WMAX(B)=\{b_3,b_4,b_5,b_7,b_8\}$. Indeed, the vectors
$b_1,b_2$ and $b_6$ do not belong to the weakly Pareto curve since
neither $t_1$, nor $t_2,t_3,t_4$ is dominated by them. In contrary,
$b_3$ and $b_4$ belong to the weakly Pareto curve since they
dominate $t_3$.
}
\end{center}
\end{figure}
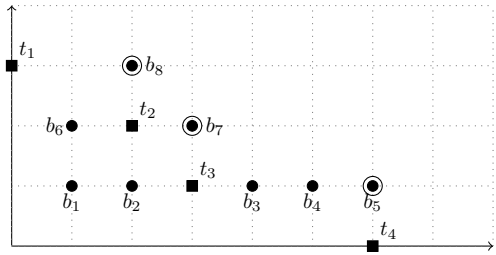


\newpage

\appendix

\section{Appendix}
In order to prove Proposition~\ref{prop-ref}, we first need a lemma.

\begin{lemma}\label{lem-p1}
Given $t \in \Omega_S$, one has
$\big( \forall s \in S, \, s \not \gg t \big) \Leftrightarrow \big( \forall s \in MAX(S), \, s \not \gg t \big)$.
\end{lemma}
\begin{proof}
\noindent ($\Rightarrow$) Obvious since $MAX(S) \subseteq S$.
\noindent ($\Leftarrow$) Obvious if $S=MAX(S)$, so let assume that
$MAX(S)\subset S$ and let $s \in S \setminus MAX(S)$.
By the definition of $MAX$, there exists $x \in MAX(S)$
such that $x \succeq s$. Therefore $x^j \ge s^j$ for all $j \in K$. 
Since $x \not \gg t$ (by hypothesis) there is $\ell \in K$ such that 
$x^{\ell} \le t^{\ell}$. We deduce that $s^\ell \le t^{\ell}$ and 
$s \not \gg t$ follows.
\ \\$\Box$
\end{proof}

\medskip
\noindent {\bf Proof of Proposition~\ref{prop-ref}:}\\
For any finite set $S\subset \mathbb{R}^k_{*+}$, 
and any $v \in \mathbb{R}^k_{*+} \setminus MAX(S)$, 
we want to proove the following equivalence:
\[ \exists t \in {\cal T}(S) \, : \,  v \gg t \Longleftrightarrow 
v \in MAX(S \cup \{v\}), \]
with ${\cal T}(S) := MIN( \{t \in \Omega_S \, : \, 
\forall s \in S, \; s \not \gg t \}).$\\
$(\Rightarrow)$ Since $v \gg t$, $v^i > t^i$ for all $i \in K$. 
Since $t \in {\cal T}(S)$, $s \not \gg t$ for all 
$s \in MAX(S)$. It means that for all
$s \in MAX(S)$ there exists (at least) one index, 
says $i(s)$, such that $s^{i(s)} \le t^{i(s)}$. 
Then, for all $s \in MAX(S)$, $s^{i(s)} \le t^{i(s)} 
< v^{i(s)}$ holds. It follows that $s \not \succ v $ for all 
$s \in MAX(S)$. Thus, $v \in MAX(S \cup
 \{v\})$.

\medskip

\noindent $(\Leftarrow)$ $v \in MAX(S \cup \{v\})$ 
means that for all $s \in MAX(S)$ either 
$v \sim s$ or $v \succeq s$. Since $v \in 
\mathbb{R}^k_{*+} \setminus MAX(S)$, the case 
$v \succeq s$ becomes $v \succ s$. Both cases imply that, for all 
$s \in MAX(S)$, there exists an index, says 
$i(s)$, such that $v^{i(s)} > s^{i(s)}$. 
Given $a \in \mathbb{R}_{*+}$, let us denote by
$dec_{S,j}(a)$ the quantity $\max\{x \in D_S^j \, : \, x<a\}$, with
$D_S^j:=\{0\} \cup \{s^j \, : \, s\in S\}$ for $j=1,\ldots ,k$.
Now let us consider a 
vector $\omega$ satisfying $\omega^i = dec_{S,i}(v^i)$ for all $i \in K$.
By definition of $dec_{S,i}$, one has $\omega^{i(s)} \ge s^{i(s)}$ for 
all $s \in MAX(S)$. Therefore, 
$s \not \gg \omega$  for all $s \in MAX(S)$. We 
deduce that $\omega \in \{ t \in \Omega_S \, : \, \forall s \in 
MAX(S), \; s \not \gg t \}$. 
By definition of a minimum operator $MIN$,
there exists $q \in MIN( \{ t \in \Omega_S \, : \, 
\forall b \in MAX(S), \; s \not \gg t \})$ such that 
$q \preceq \omega$. 
Using Lemma~\ref{lem-p1}, the set $\{ t \in \Omega_S \, : \, \forall s \in 
MAX(S) \; s \not \gg t \}=\{ t \in \Omega_S \, : \,
 \forall s \in S, \; s \not \gg t \}$. By construction,
$v \gg \omega$ holds. 
Finally, $v  \gg \omega \succeq q$ implies $v \gg q$ where 
$q \in MIN(\{ t \in \Omega_S \, : \, \forall s \in S,
 \; s \not \gg t \})={\cal T}(S)$.
\ \\$\Box$

For $s\in \mathbb{R}^k_{*+}$ and $1\leq j\leq k$, let us define $s(j)$ as the vector 
which has its $j$-th coordinate equals to $s^j$, and 0 elsewhere.
Given a set of vectors $v_1,\ldots ,v_n\in \mathbb{R}^k$, let us define
$VMAX(\{v_1,\ldots ,v_n\})$ as the vector obtained by taking on each coordinate
$j\in \{1,\ldots ,k\}$ the maximum value among the $j$-th coordinate of
$v_1,\ldots ,v_n$, i.e.
$VMAX(\{v_1,\ldots ,v_n\}) = (\max_{1\leq i\leq n} v_i^1, \ldots 
, \max_{1\leq i\leq n} v_i^k)$.

\begin{proposition} \label{prop-ref2}
For any finite set $S\subset \mathbb{R}^k_{*+}$, one has\\
${\cal T}(S) = MIN(\cup_{\{j_s\in \{1,\ldots ,k\} | 
 s\in MAX(S)\}} VMAX(\cup_{s\in MAX(S)} s(j_s))).$
\end{proposition}
\begin{proof}
Let us first describe another way to compute the set of reference points which will
be useful to prove some properties on it.
Let $S=\{s_1,\ldots ,s_n\} \subset \mathbb{R}^k_{*+}$, and let us consider 
a vector $v\in \mathbb{R}^k_{*+}\setminus MAX(S)$ such that 
$v\in MAX(S\cup \{v\})$. Of course $v\not\in S$, and
it means that $\forall s\in S$, $v$ is not dominated by $s$, i.e.
$s\not \succ v$ and $s\neq v$, i.e. $\exists j$, $1\leq j\leq k$, such that 
$v^j>s^j$.
We define a boolean formula $\psi$ indexed by the vector $s$, in the 
following way: $\psi(s) := (v\gg s(1)) \vee (v\gg s(2)) \vee 
\ldots \vee (v\gg s(k))$.
Since we know that $v\gg 0$, we get that $v$ is not dominated by $s$ if and
only if the boolean formula $\psi(s)$ is true.
Finally, observe that $\forall s\in S$, $v$ is not dominated by $s$, is
equivalent with $\forall s\in MAX(S)$, $v$ is not dominated by $s$.
Therefore, one has $v\in MAX(S\cup \{v\})$ if and only if
$\psi := \bigwedge_{s\in MAX(S)} \psi(s)$ is true, i.e.
$$\psi := \bigwedge_{s\in MAX(S)} 
  [ (v\gg s(1)) \vee (v\gg s(2)) \vee \ldots \vee (v\gg s(k)) ]$$
is true.

We can rewrite $\psi$ as

$$\psi := \bigvee_{\{j_s\in \{1,\ldots ,k\} | s\in MAX(S)\}} \
  \bigwedge_{s\in MAX(S)} (v\gg s(j_s)).$$

In this formula, the $\vee$ is taken over all tuples
$(j_s)_{s\in MAX(S)}$, with $j_s\in \{1,\ldots ,k\}$ for 
each $s\in MAX(S)$, there are therefore 
$k^{|{\cal P}_{max}(S)|}$ such tuples.

Clearly, one has $\bigwedge_{s\in MAX(S)} (v\gg s(j_s)) \iff v\gg VMAX(\cup_{s\in MAX(S)} 
s(j_s))$, therefore we can write
$$\psi := \bigvee_{\{j_s\in \{1,\ldots ,k\} | s\in MAX(S)\}} 
(v\gg VMAX(\cup_{s\in MAX(S)} \: s(j_s))).$$

Now observe that if $v_1\succ v_2$, clearly one has
$v\gg v_2$ if and only if $(v\gg v_1) \vee (v\gg v_2)$.
We therefore have proved the proposition.
\ \\$\Box$
\end{proof}

\medskip
\noindent {\bf Proof of Lemma~\ref{lem-la} and \ref{lem-lt}:}\\
It is a direct consequence of Proposition~\ref{prop-ref2}.

\medskip
\noindent {\bf Proof of Lemma~\ref{lem-t01}:}\\
Let $s\in MAX(S)$. Then $\forall s'\in MAX(S)$,
$s' \neq s$, $\exists l_{s'}$,
$1\leq l_{s'}\leq k$, such that
$(s')^{l_{s'}} < s^{l_{s'}}$. If we take $j_{s'}=l_{s'}$ for $s'\neq s$, and
$j_s$ any integer value between 1 and $k$, we have
$\forall s'\in MAX(S)$, $s\succ s'(j_{s'})$, and therefore
$s\succ VMAX(\cup_{s'\in MAX(S)} s'(j_{s'}))$.
The result now follows from Proposition~\ref{prop-ref2}.
\ \\$\Box$


\bibliographystyle{named}
\bibliography{ijcai19}

\end{document}